\documentclass[10pt,conference]{IEEEtran}
\usepackage{amsthm}
\usepackage{amsmath}
\usepackage{amssymb}
\usepackage{amsfonts}

\usepackage{listings}
\usepackage{etoolbox}

\renewcommand\thesection{\arabic{section}}
\renewcommand\thesubsection{\thesection.\arabic{subsection}}
\renewcommand\thesubsubsection{\thesubsection.\arabic{subsubsection} }

\renewcommand\thesectiondis{\arabic{section}}
\renewcommand\thesubsectiondis{\thesectiondis.\arabic{subsection}}
\renewcommand\thesubsubsectiondis{\thesubsectiondis.\arabic{subsubsection}}
\newcommand*{\defeq}{\stackrel{\text{def}}{=}}

\providebool{techreport}
\setbool{techreport}{false}

\usepackage{array}
\usepackage{graphicx}
\usepackage{cite}
\usepackage{setspace}
\usepackage[justification=Centering]{subfig}
\usepackage{algorithm}
\usepackage{algorithmic}
\usepackage{tikz}
\usepackage{xcolor}
\usepackage{multirow}
\usepackage{paralist}
\usepackage{url}
\usepackage{etoolbox}
\usepackage{cancel}

\usepackage{color}

\newcommand{\ceiling}[1]{\left\lceil{#1}\right\rceil}

\newcommand{\setof}[1]{\left\{{#1}\right\}}

\newcommand{\defn}[1]       {{\textit{\textbf{\boldmath #1}}}}

 \def\myendproof{{\ \vbox{\hrule\hbox{%
   \vrule height1.3ex\hskip0.8ex\vrule}\hrule }}\par}
 
 \setboolean{ALC@noend}{true}
\providebool{techreport}
\setbool{techreport}{true}

\newtheorem{theorem}{Theorem}
\newtheorem{lemma}{Lemma}

\newtheorem{example}{Example}

\newtheorem{definition}{Definition}

\newtheorem{observation}{Observation}

\graphicspath{{fig/multiframe/}{fig/arbitrary/}{fig/dag/}{figures/}} 

\usetikzlibrary{%
  arrows,%
  shapes.misc,
  shapes.arrows,%
  chains,%
  matrix,%
  positioning,
  scopes,%
  decorations.pathmorphing,
  shadows%
}

\pgfdeclarelayer{background}
\pgfdeclarelayer{foreground}
\pgfsetlayers{background,main,foreground}
\tikzstyle{materia}=[draw, fill=white, text width=1.0em, text centered,
  minimum height=4em,drop shadow]
\tikzstyle{practica} = [materia, text width=18em, minimum width=8em,
align =left,
  minimum height=3em, rounded corners, drop shadow]
\tikzstyle{texto} = [above, text width=6em, text centered]
\tikzstyle{linepart} = [draw, thick, color=blue!50, -latex', dashed]
\tikzstyle{line} = [draw, line width = 2pt, color=blue!50, -latex']
\tikzstyle{ur}=[draw, text centered, minimum height=0.01em]

\newcommand{\georg}[1]{}

\title{Reservation-Based Federated Scheduling for  Parallel Real-Time Tasks}
\vspace{3mm}

\begin{document}

\author{Niklas Ueter$^1$, Georg von der Br\"uggen$^1$, Jian-Jia Chen$^1$,  Jing Li$^2$, and
  Kunal Agrawal$^3$\\
$^1$TU Dortmund University, Germany\\
$^2$New Jersey Institute of Technology, U.S.A\\
$^3$Washington University in St. Louis, U.S.A}

\maketitle
\begin{abstract}
This paper considers the scheduling of parallel real-time tasks with arbitrary-deadlines. Each job of a parallel task is described as a directed acyclic graph
(DAG).  In contrast to prior work in this area, where decomposition-based 
scheduling algorithms are proposed based on the DAG-structure and inter-task 
interference is analyzed as self-suspending behavior, this paper generalizes the federated scheduling approach.
We propose a reservation-based algorithm, called reservation-based federated scheduling, that dominates federated scheduling. We provide general constraints for the design of such systems and prove that reservation-based federated scheduling has a constant speedup factor with respect to any optimal DAG task scheduler.
Furthermore, the presented algorithm can be used in conjunction with any scheduler and scheduling analysis suitable for ordinary arbitrary-deadline sporadic task sets, i.e., without parallelism.   
\end{abstract}

\section{Introduction}
A frequently used model to describe real-time systems is with a collection of 
independent tasks that release an infinite sequence of jobs according to 
some parameterizable release pattern.
The sporadic task model, where a task $\tau_i$ is characterized by its
relative deadline $D_i$, its minimum inter-arrival time $T_i$, and its
worst-case execution time (WCET) $C_i$, has been widely adopted for 
real-time systems. 
A sporadic task is an infinite sequence of task instances, referred 
to as \emph{jobs}, where the arrival of  two consecutive jobs of a task
is separated at least by its minimum inter-arrival time.
In real-time systems, tasks must fulfill timing requirements, i.e., each job
must finish at most $C_i$ units of computation between the arrival of a job at 
$t_a$ and that jobs absolute deadline at $t_a+D_i$.
A sporadic task system ${\tau}$ is called an \defn{implicit-deadline}
system if $D_i = T_i$ holds for each $\tau_i$ in ${\tau}$, and is 
called a \defn{constrained-deadline} system if $D_i \leq T_i$
holds for each $\tau_i$ in ${\tau}$.  Otherwise, such a sporadic task system
${\tau}$ is an \defn{arbitrary-deadline} system.

Traditionally, each task $\tau_i$ is only associated with its
worst-case execution time (WCET) $C_i$, since in uniprocessor platforms 
the processor executes only one job at each point in time and there is no need
to express potential parallel execution paths.
However, modern real-time systems increasingly employ multi-processor platforms to
suffice the increasing performance demands and the need for energy efficiency.
Multi-processor platforms allow both \emph{inter-task parallelism}, 
i.e., to execute sequential programs concurrently,
and \emph{intra-task parallelism}, 
i.e., a job of a parallelized task can be executed on multiple processors at the
same time.
To enable {intra-task parallelism}, programs
are expected to be potentially executed in parallel which must be enabled by the
software design. 
An established model for parallelized tasks is the \defn{Directed-Acyclic-Graph (DAG)} model. 
Through out this paper, we consider how to schedule a sporadic DAG
task set $\tau$ on a multi-processor system with $M$ homogeneous processors.

\subsection*{Task Model}
The Directed-Acyclic-Graph (DAG) model for a parallel task
expresses intra-task dependencies and which subtasks can potentially be
executed in parallel~\cite{saifullah2014parallel,DBLP:conf/ecrts/LiALG13,DBLP:conf/ecrts/BonifaciMSW13}.
In particular, the execution of a task $\tau_i$ can be
divided into subtasks and the precedence constraints of these subtasks
are defined by a DAG structure. An example is presented
in~Figure~\ref{fig:example}. Each node represents a subtasks and the directed
arrows indicate the precedence constraints. Each node
  is characterized by the worst-case execution time of the
  corresponding subtask. 


\begin{figure}[t]
\includegraphics[width=0.5\textwidth]{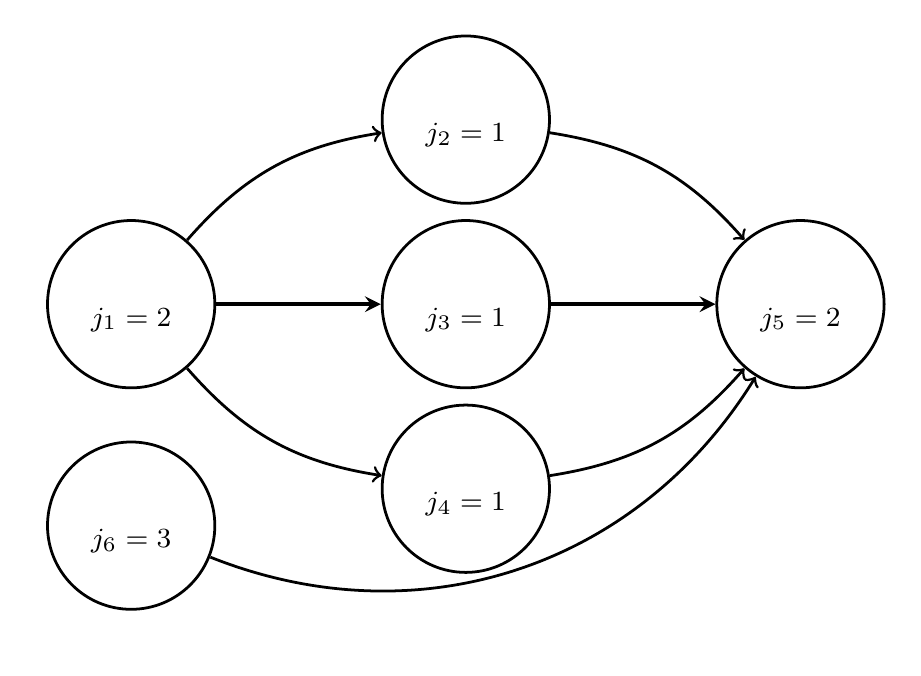}
\caption{A sporadic, constrained-deadline DAG task with \mbox{$C_i=10$}, $L_i =
5, D_i = 9, T_i = 12$.}
\label{fig:example}
\end{figure}

For a DAG, two parameters are of importance:
\begin{itemize}
\item \defn{total execution time} (or \defn{work}) $C_i$ of task $\tau_i$:
  the summation of the worst-case execution times of all the
  subtasks of task $\tau_i$.  
\item \defn{critical-path length} $L_i$ of task $\tau_i$: 
  the length of the critical path in the given DAG, i.e., 
  the worst case execution time of the task on an infinite number of
  processors. 
\end{itemize}
By definition, $C_i \geq L_i > 0$ for every task $\tau_i$. The \defn{utilization} 
  of task $\tau_i$ is denoted by $U_i=\frac{C_i}{T_i}$. 

This way of parametrization has the advantage 
to be completely
agnostic of the internal parallelization structure, i.e., how many
sub-tasks exist and how the precedence constraints amongst them
are. Scheduling algorithms that can feasibly schedule DAG task sets
solely  based on these two parameters
also 
allow the change
of the DAG structure during runtime as long as those constraints are
met.  The apparent downside of this abstraction is the pessimism, since the
worst possible structure has to be considered regardless of the actual 
structure, and the scheduling algorithms have to suffice the tasks deadline for
all possible structures under given parameter constraints.

\subsection*{Related work}

The scheduling of parallel real-time DAG tasks
has been widely researched in various directions.
To the best of our knowledge, 
three general scheduling approaches exist:
\begin{itemize}
\item \emph{No treatment}: The DAG structure and parameters of a
   task are not utilized or used at all for
  scheduling decisions. Whenever a subtask of task $\tau_i$ is ready to be
  executed, the standard global or partitioned multiprocessor
  scheduling is used to schedule the subtasks, e.g., 
  \cite{DBLP:conf/opodis/AnderssonN12,DBLP:conf/ecrts/BonifaciMSW13,DBLP:conf/ecrts/LiALG13,Li:ECRTS14}.

\item \emph{Decomposition-based strategies}: A DAG task is decomposed
  into a set of sequential tasks with specified relative deadlines and
  offsets of their release times.  These sequential tasks are then
  scheduled accordingly without considering the DAG structure anymore,
  e.g.,
  \cite{Lakshmanan:2010:SPR:1935940.1936239,DBLP:conf/rtss/SaifullahALG11,DBLP:conf/iccps/KimKLR13,DBLP:conf/ecrts/NelissenBGM12,DBLP:conf/rtss/JiangLGW16,saifullah2014parallel}.
  Decomposition-based strategies utilize the DAG structure off-line in
  order to apply the decomposition.
\item \emph{Federated scheduling}: The task set is partitioned into
  light and heavy tasks.  Light tasks are those, that can be
  completely sequentialized and still fit on one processor.  
  On the other hand, a task that needs 
  more than one processor to meet its deadline is a heavy task. In the
  original design of federated scheduling for implicit-deadline task
  systems proposed by Li et al.~\cite{Li:ECRTS14}, a light task is
  solely executed \emph{sequentially} without exploiting the
  parallelized structure, and a heavy task is assigned to its
  designated processors that \emph{exclusively} execute only the heavy
  task. Baruah~\cite{DBLP:conf/date/Baruah15,DBLP:conf/ipps/Baruah15,DBLP:conf/emsoft/Baruah15}
  adopted the concept of federated scheduling for scheduling
  constrained-deadline and arbitrary-deadline task
  systems. Chen~\cite{Chen:2016:FSA:2990336.2990421} later showed that
  federated scheduling does not admit any constant speedup factor with
  respect the optimal scheduling algorithm. Jiang et
  al.~\cite{RTSS-2017-semi-federated} extended the federated
  scheduling approach to semi-federated scheduling,
  in which one or two processors used by a heavy task can be shared with other tasks.
\end{itemize}

\subsection*{Contributions}
A downside of {federated scheduling} is the granting of
processors to heavy tasks exclusively, thus alleviating the potential
to map light tasks onto the same processors. To address these limitations, 
this paper provides the following results:
\begin{itemize}
\item {We propose a
\defn{reservation-based federated scheduling} for DAG tasks that provides provably sufficient
amount of service for each DAG task to meet its relative deadline and
provides a simple, timing isolated interface for analysis.
That means, the DAG task can be treated like an arbitrary- or constrained-deadline, 
sporadic real-time task analytically.
Hence we show how to reduce the problem of scheduling sporadic, arbitrary-deadline 
DAG tasks to the problem of scheduling sequential sporadic, arbitrary-deadline tasks.}

\item {Specifically, we provide algorithms to transform a set of sporadic, arbitrarily-deadline 
DAG tasks into a set of sequential sporadic, arbitrary-deadline real-time tasks that 
can be scheduled by using any scheduling algorithm, that supports the aforementioned task model.}

\item {Moreover, we provide general design rules and constraints for providing provably 
sufficient and heuristically good reservations for use in Partitioned (and Global)
Scheduling algorithms.}

\item {We further resolve the problem of non-constant speedup factors of federated scheduling for arbitrary-deadline 
DAG task sets with respect to any optimal scheduling algorithm that was pointed
out by Chen~\cite{Chen:2016:FSA:2990336.2990421}.
We show, that this speedup factor is at most $2+\sqrt{3}$ by the setting of 
a specific workload inflation.}
\end{itemize}

\section{Issues of Federated Scheduling for Constrained-Deadline Systems}

Here, we reuse the example presented by
Chen~\cite{Chen:2016:FSA:2990336.2990421} to explain the main issue of
applying federated scheduling for constrained-deadline task systems.
Suppose that $M \geq 2 $ is a positive integer. Moreover, let $K$ be
any arbitrary number with $K \geq 2$.  We create $N$
constrained-deadline sporadic tasks with the following setting:
\begin{itemize}
\item $C_1=M$, $D_1=1$, and $T_1\rightarrow \infty$.
\item $C_i=K^{i-2} (K-1)M$, $D_i=K^{i-1}$, and $T_i=\infty$ for $i=2,3,\ldots, N$.
\end{itemize}
Table~\ref{tab:example} provides a concrete example for $N=10$, $M=10$
and $K=2$.
Each task $\tau_i$ has $M$ subtasks, there is no precedence
constraint among these $M$ subtasks (which is a special case of DAG), and each
subtask of task $\tau_i$ has the worst-case execution time of
$\frac{C_i}{M}$.

An obviously feasible schedule is to assign each subtask of task
$\tau_i$ to one of the $M$ processors. However, as task $\tau_1$ can
only be feasibly scheduled by running on all the $M$ processors in
parallel, federated scheduling \emph{exclusively}
allocates all the $M$ processors to task $\tau_1$. Similarly, the
semi-federated scheduling in \cite{RTSS-2017-semi-federated} also
suffers from such exclusive allocation.

\begin{table}[t]
  \centering
  \begin{tabular}{|c||c|c|c|c|c|c|c|c|c|c|c|}
  \hline
   & $\tau_1$   & $\tau_2$  & $\tau_3$   & $\tau_4$   & $\tau_5$   & $\tau_6$   & $\tau_7$   & $\tau_9$   & $\tau_{10}$\\
   \hline
 $C_i$ & 10 & 10 & 20 & 40 & 80 & 160 & 320 & 640 & 1280\\
\hline
 $D_i$ & 1  & 2   & 4   & 8    & 16 & 32  & 64 & 128 & 256\\
  \hline    
 $T_i$ & \multicolumn{9}{c|}{$\rightarrow \infty$}\\
  \hline    
  \end{tabular}
  \caption{An example of the task set $\tau$ when $N=10$, \mbox{$M=10$}, and
  $K=2$, from \cite{Chen:2016:FSA:2990336.2990421}}
  \label{tab:example}
\end{table}

From this example, we can see that the main issue of applying federated scheduling for
constrained-deadline task systems is the exclusive allocation of heavy
tasks. Such a heavy task may need a lot of processors due to its short
relative deadline, but have very low utilization in the long run if
its minimum inter-arrival time is very long. Allocating many
processors to such a heavy task results in a significant waste of resources.

Our proposed approach in
this paper is to use reservation-based allocation instead of exclusive
allocation for heavy tasks. Therefore, instead of dedicating a few
processors to a heavy task, we assign a few reservation servers to a
heavy task. The timing properties of a DAG task will be guaranteed as
long as the corresponding reservations can be guaranteed to be
feasibly provided. We will detail the concept in the next section.

\section{Reservation-Based Federated Scheduling}

An inherent difficulty when analyzing the schedulability of DAG task systems is 
the intra-task dependency in conjunction with the inter-task interference.
Federated scheduling avoids this problem by granting a subset of available
processors  to heavy tasks exclusively and therefore avoiding inter-task interference.
A natural generalization of the federated scheduling approach is to reserve
sufficient  resources to heavy tasks exclusively.
This approach combines the advantage of avoiding inter-task interference and
self-suspension  with the possibilities to fit the amount of resources required more precisely.
The reservation-based federated approach requires to quantify the maximum
computation demand a DAG task can generate 
over any interval and
quantify the sufficient amount of resources during that interval.

\subsection{Basic Concepts}

In this paper, we enforce the reservations to be provided
synchronously with the release of a DAG task's job. This means,  
whenever a DAG task releases a job at $t_0$, the 
associated service is provided during the release- and 
deadline-interval $[t_0,t_0+D_i)$.
In order to provide a well known interface, the service providing 
reservations are modeled as an ordinary sporadic, arbitrary-deadline task 
more formally described in the following definition.

\begin{definition}
A reservation generating sporadic task $\tau_{i,j}$ for serving a DAG task
$\tau_i$  is defined by the tuple $(E_{i,j},D_i,T_i)$, such that $E_{i,j}$ is the
amount of  computation reserved over the interval $[t_0,t_0+D_i)$ with a 
minimum inter-arrival time of $T_i$.
\end{definition}


Over an interval of $[t_0, t_0+D_i)$ where $t_0$ denotes the release
of a job of the DAG task $\tau_i$, we create $m_i$ instances (jobs) of
sporadic real-time \emph{reservation servers} released with execution
budgets $E_{i,1},E_{i,2},..,E_{i,m_i}$ and relative deadline $D_i$,
that are scheduled according to some scheduling algorithm on a
homogeneous multiprocessor system with $M$ processors.
Moreover, the jobs that are released at time $t_0$ by the reservation
servers are only used to serve the DAG job of task $\tau_i$ that  arrived at
time $t_0$. Especially, they are not used to serve any other jobs of task
$\tau_i$ that arrived after $t_0$.  The operating system can apply any scheduling
strategy to execute the $m_i$ instances. If an instance of a reservation server
reserved for task $\tau_i$ is executed at time $t$, we say that the
system provides (or alternatively the reservation servers provide)
service to run the job of task $\tau_i$ arrived at time $t_0$. 
On the other hand, the $m_i$ reservation servers do not provide any service at
time $t$ if none of them is executed at time $t$ by the scheduler.

The scheduling algorithm for DAG is \emph{list scheduling}, which is
workload-conserving with respect to the service provided by the
reservation servers.  Namely, at every point in time in which the DAG task has
pending workload and the system provides service (to run a reservation server), the
workload is executed.

In conclusion, the problem of scheduling DAG task sets and the analysis thereof is hence 
divided into the following two problems:
\begin{enumerate}
  \item Scheduling of sporadic, arbitrary-deadline task sets.
  \item Provide provably sufficient reservation to service a set of arbitrary DAG tasks.
\end{enumerate}

\begin{theorem}
  \label{thm:response-time-bound-one-job}
  Suppose that $m_i$ sequential instances (jobs) of real-time reservation servers
  are created and released for serving a DAG task $\tau_i$ with
  execution budgets $E_{i,1},E_{i,2},..,E_{i,m_i}$ when a job of task
  $\tau_i$ is released at time $t_0$. The job of task $\tau_i$ arrived
  at time $t_0$ can be finished no later than its absolute deadline
  $t_0+D_i$ if
  \begin{itemize}
  \item \textbf{[Schedulability Condition]}: the $m_i$ sequential jobs of the reservation servers can
    be guaranteed to finish no later than their absolute deadline at $t_0+D_i$, and
  \item \textbf{[Reservation Condition]}: $C_i + L_i \cdot (m_i-1) \leq \sum_{j=1}^{m_i} E_{i,j}$.
  \end{itemize}
\end{theorem}
\begin{proof}
  We consider an arbitrary execution schedule $S$ of the $m_i$
  sequential jobs executed from $t_0$ to $t_0+D_i$.  Suppose,
  \textit{for contradiction}, that the reservation condition holds but there is an unfinished subjob of
  the DAG job of task $\tau_i$ at time $t_0+D_i$ in 
  $S$. 
  Since the list scheduling algorithm is applied, the schedule for a
  DAG job is under a certain topological order and is
  workload-conserving. That is, unless a DAG job has finished at time
  $t$, whenever the system provides service to the DAG job, one of its
  subjobs is executed at time $t$.

  We define the following terms based on the execution of the DAG job
  of task $\tau_i$ arrived at time $t_0$ in the schedule $S$.  Let the
  last moment prior to $t_0+D_i$ when the system provides service to
  the DAG job be $f_\ell$ in the schedule $S$. Moreover, $c_\ell$ is a
  subjob of task $\tau_i$ executed at $f_\ell$ in 
  $S$. Let $\theta_\ell$ be the earliest time in 
  $S$ when
  the subjob $c_\ell$ is executed. After $\theta_\ell$ is determined,
  among the predecessors of $c_\ell$, let the one finished
  \emph{last} in the schedule $S$ be $c_{\ell-1}$. Moreover, we
  determine $f_{\ell-1}$ as the finishing time of $c_{\ell-1}$ and
  $\theta_{\ell-1}$ as the starting time of $c_{\ell-1}$ in the
  schedule $S$. By repeating the above procedure, we can define
  $\theta_1, f_1, c_1$, where there is no predecessor of $c_1$ any
  more in 
  $S$. For notational brevity, let $f_0$ be
  $t_0$.

  According to the above construction, the sequence $c_1, c_2, \ldots,
  c_{\ell}$ is a \emph{path} in the DAG structure of $\tau_i$. Let
  $exe(c_j)$ be the execution time of $c_j$. By definition, we know
  that $\sum_{j=1}^{\ell} exe(c_j) \leq L_i$.  In the schedule 
  $S$,
  whenever $c_j$ finishes, we know that $c_{j+1}$ can be executed, but
  there may be a gap between $f_j$ and $\theta_{j+1}$. 

  Suppose that $\beta_i(x, y, S)$ is the accumulative amount of
  service provided by the $m_i$ sequential jobs in an interval $[x,
  y)$ in 
  $S$. 
  Since the list scheduling
  algorithm is workload-conserving, if $c_j$ is not executed at time
  $t$ where $\theta_j \leq t \leq f_j$, then all the services are used
  for processing other subjobs of the DAG job of task
  $\tau_i$. Therefore, for $j=1,2,\ldots,\ell$, the maximum amount of
  service that is \emph{provided to the DAG job but not used} in time
  interval $[\theta_j, f_j)$ in 
  $S$ is at most $(m_i-1)
  exe(c_j)$, since each of the $m_i$ reservation servers can only
  provide its service sequentially. That is, in the interval
  $[\theta_j, f_j)$ at least $\max\{\beta_i(\theta_j, f_j, S) -
  exe(c_j) \times (m_i-1), exe(c_j)\}$ amount of execution time of the
  DAG job is executed.

  Similarly, for $j=1,2,\ldots,\ell$, the maximum amount of service
  that is \emph{provided to the DAG job but not used} in time interval
  $[f_{j-1}, \theta_j)$ in 
  $S$ is $0$; otherwise $c_j$
  should have been started before $\theta_j$. Therefore, in the interval
  $[f_{j-1}, \theta_j)$ at least $\beta_i(f_{j-1}, \theta_j, S)$
  amount of execution time of the DAG job is executed.

  Under the assumption that the job misses its deadline at time
  $t_0+D_i$ and the $m_i$ sequential jobs of the reservation servers
  can finish no later than their absolute deadline at $t_0+D_i$ in the
  schedule $S$, we know that 

{\footnotesize \begin{align*}
    & C_i\\
     > &  \sum_{j=1}^{\ell} \beta_i(f_{j-1}, \theta_j, S) + \max\{\beta_i(\theta_j, f_j, S) - exe(c_j)  (m_i-1), exe(c_j)\}\\
      \geq & \sum_{j=1}^{\ell} \beta_i(f_{j-1}, \theta_j, S) + \beta_i(\theta_j, f_j, S) - exe(c_j)  (m_i-1)\\
      = & \sum_{j=1}^{m_i} E_{i,j} - \sum_{j=1}^{\ell}(m_i-1)\times exe(c_j)\\
      = & \sum_{j=1}^{m_i} E_{i,j} - (m_i-1)\times L_i \geq C_i
  \end{align*}}
  Therefore, we reach the contradiction. 
\end{proof}

\subsection{Reservation Constraints}

According to Theorem~\ref{thm:response-time-bound-one-job}, we should
focus on providing the reservations such that $C_i + L_i(m_i-1) \leq
\sum_{j=1}^{m_i} E_{i,j}$. The following lemma shows that any
reservation with $E_{i,j} < L_i$ has no benefit for meeting such a condition.

\begin{lemma}
  \label{lemma:Eij>Li}
  If there exists a $\tau_{i,j^*}$ with $E_{i,j^*} < L_i$, such a
  reservation $\tau_{i,j}$ has a negative impact on the condition
  $\sum_{j=1}^{m_i} E_{i,j} - (C_i + L_i(m_i-1) )$.
\end{lemma}
\begin{proof}
  This comes from simple arithmetic. If so, removing the reservation
  $\tau_{i,j^*}$ leads to $m_i-1$ reservation servers with better
  reservations due to $\sum_{j=1}^{m_i} E_{i,j} - (C_i + L_i(m_i-1)) <
  (\sum_{j=1}^{m_i} E_{i,j}) - E_{i,j^*} - (C_i + L_i(m_i-2))$.
\end{proof}

Therefore, we will implicitly consider the property in
Lemma~\ref{lemma:Eij>Li}, i.e., $E_{i,j} \geq L_i, \forall j$ whenever
the reservation condition in
Theorem~\ref{thm:response-time-bound-one-job} is used.  For further
analysis let $E_{i,j} \defeq \gamma_{i,j} \cdot L_i$, with $1 < \gamma_{i,j}
\leq \frac{D_i}{L_i}$ and therefore any reservation system $\mathcal{S} \defeq
(m_i,\gamma_{i,1},\gamma_{i,2},..,\gamma_{i,j})$, that suffices the following
constraints
\begin{subequations}
\label{eq:lp-condition-overall}
\begin{align}
L_i \cdot (m_i-1) + C_i \leq \sum_{j=1}^{m_i}\gamma_{i,j} \cdot L_i \label{eq:over-all-reservation}\\
\gamma_{i,j} \cdot L_i \leq D_i \ \forall 1 \leq j \leq m_i \\
\gamma_{i,j} > 1 \ \forall 1 \leq j \leq m_i
\end{align} 
\end{subequations} 
is feasible for satisfying the reservation condition in
Theorem~\ref{thm:response-time-bound-one-job}. 

The cumulative reservation budget to serve a DAG task is given by
\begin{equation}
\label{eq:cumulative-budget}
C_i' = \sum_{j=1}^{m_i} E_{i,j} = L_i \cdot \sum_{j=1}^{m_i} \gamma_{i,j}.
\end{equation}

In the special case of \emph{equal-reservations}, a lower bound 
of the required amount of reservations can be solved analytically to 
\begin{equation}
\label{eq:budget-constraints}
L_i \cdot (m_i-1)+C_i\leq \gamma_i \cdot m_i \cdot L_i,
\end{equation} which yields
\begin{equation}
\label{eq:mi-lowerbound}
\frac{C_i-L_i}{L_i \cdot (\gamma_i-1)} \leq m_i.
\end{equation}

Note that the notation of $\gamma_{i,j}$ changed to $\gamma_i$, due to 
equal size for all $1 \leq j \leq m_i$. Since the amount of reservations 
must be a natural number we know that
\begin{equation}
\label{eq:mi-def}
\ceiling{\frac{C_i-L_i}{L_i \cdot (\gamma_i-1)}} \defeq m_i
\end{equation} and that $m_i$ is the smallest amount of reservations required if 
all reservation-budgets are equal in size.
Additionally, due to the fact that there are instances in which multiple
settings of $\gamma_i$ yield the same minimal amount of reservations, we define 
\begin{equation}
\label{eq:gamma-min}
\gamma_i = \min\lbrace \gamma_i \ | \ \gamma_i \ \text{satisfies Eq.~\eqref{eq:mi-def}} \rbrace.
\end{equation}

\begin{observation}
The left-hand side of the above
equation~\eqref{eq:mi-def} is minimised, if $\gamma_i$ is maximised, i.e., 
\begin{equation}
m_i = \ceiling{\frac{C_i-L_i}{D_i-L_i}} \nonumber
\end{equation} and the corresponding smallest $\gamma_i$, that achieves an 
equally minimal amount of reservations is given by $1+\frac{C_i-L_i}{m_i L_i}$.
\label{observ1}
\end{observation}

This observation motivates the idea behind the \emph{transformation 
algorithm} \defn{R-MIN}, whose properties are described in the following theorem.


\begin{theorem}
The \emph{R-MIN} algorithm (c.f. Alg.~\ref{alg:r-min}) transforms a set of sporadic, arbitrary-deadline 
DAG tasks into a set of sporadic, arbitrary-deadline sequential tasks, that provide sufficient 
resources to schedule their associated DAG tasks.  
\myendproof
\end{theorem}

Intuitively, R-MIN classifies tasks into light and heavy tasks. For each heavy
task,  it assigns the minimum number of reservation servers to the task and
calculates  the minimum equal-reservations for  servers based on Observation~\ref{observ1}.

\begin{figure}[t]
\includegraphics[width=0.5\textwidth]{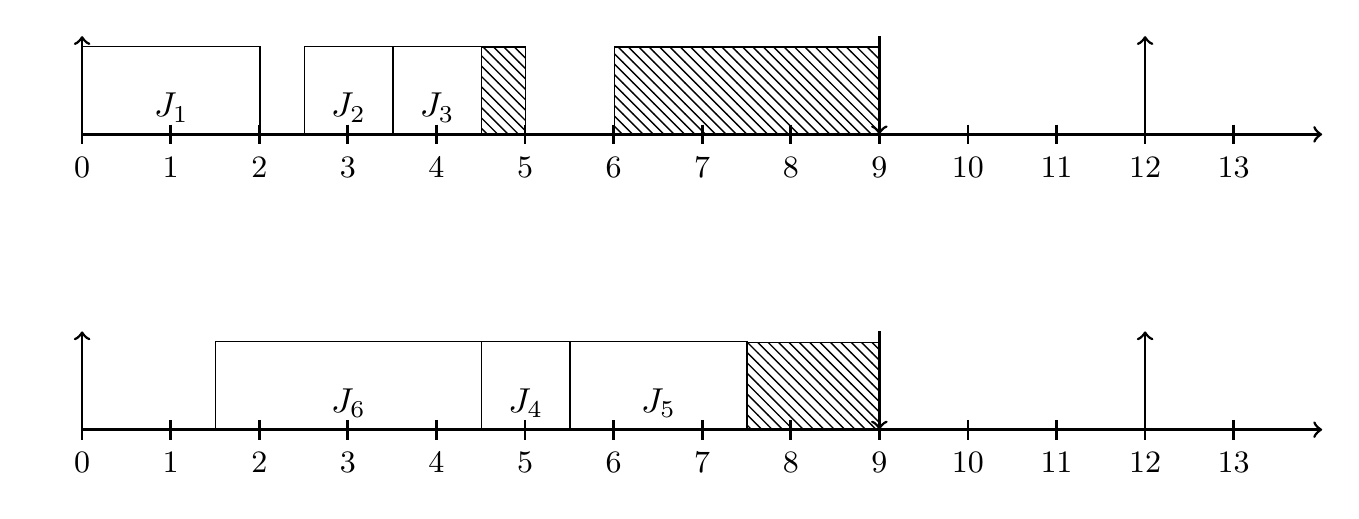}
\caption{An arbitrary schedule of two equal reservations, as computed by the
\textsc{r-min} algorithm.
The DAG task shown in Fig.~\ref{fig:example} is scheduled according to the \emph{list-scheduling} 
algorithm by any reservation server that does not service an unfinished job at that time.
$7.5$ units of time are provided over the interval
$[0, 9)$ by each reservation, scheduled on two processors. 
The hatched areas denote a \emph{spinning} reservation whereas the white areas imply 
that the reservation is either preempted or inactive.} 
\label{fig:example-schedule}
\end{figure}


\begin{example}
To illustrate the proposed concept, an arbitrary schedule of two identical 
reservations is shown in Figure~\ref{fig:example-schedule},
servicing the DAG task in Figure~\ref{fig:example}.
The schedule provides the minimal amount of identical 
reservations and associated budgets that are required to service the given DAG task 
under any preemption pattern as determined by the \textsc{r-min} algorithm.
Over the interval $[0,9)$, 7.5 units of time are provided by each reservation 
to service the DAG task using \emph{list-scheduling}.
The hatched areas denote that the reservation spins due to the lack 
of pending jobs whereas the white gaps denote that the reservation is 
either preempted or inactive.
The amount of time the reservations are \emph{spinning} may seem overly 
pessimistic, but note that this depends on the dependencies on the preemption 
patterns and the structure of the DAG task itself. 
Thus this approach trades resources for robustness with respect to 
preemption and structure uncertainty.
\end{example}

\begin{algorithm}[t]
\caption{R-MIN Algorithm}
\label{alg:r-min}
\begin{algorithmic}[1]
\STATE ${\tau}_{\textsc{heavy}} \leftarrow \{\tau_i \in {\tau} \ | \ C_i > \min(T_i, D_i) \}$
\STATE ${\tau}_{\textsc{light}} \leftarrow \{\tau_i \in {\tau} \ | \   C_i \leq \min(T_i, D_i)\}$
\STATE ${\tau} \leftarrow {\tau}_{\textsc{light}}$
\FOR{each task $\tau_i \in {\tau}_{\textsc{heavy}}$}
\STATE $m_i \leftarrow \ceiling{\frac{C_i-L_i}{D_i-L_i}}$
\FOR{$1 \leq j \leq m_i$}
\STATE $E_{i,j} \leftarrow \left(1+\frac{C_i-L_i}{m_i L_i}\right) \cdot L_i$
\STATE $\tau_{i,j} \leftarrow  (E_{i,j}, D_i, T_i)$
\STATE ${\tau} \leftarrow {\tau} \cup \setof{\tau_{i,j}}$
\ENDFOR
\ENDFOR
\RETURN ${\tau}$
\end{algorithmic}
\end{algorithm}

Note that there are more feasible configurations to serve a DAG task as long as the conditions in 
Eq.~\eqref{eq:lp-condition-overall} are met. Non-equal reservation budgets, e.g.,
at least one reservation budget in $E_{i,1}, E_{i,2},..,E_{i,m_i}$ is different
from the others, can potentially improve schedulability in partitioned or
semi-partitioned scheduling.
This is due to the fact that variability in reservation budgets 
can be helpful in \emph{packing} them onto the available processor clusters whilst satisfying 
capacity metrics.

In order to retrieve those \emph{non-equal} reservation budgets, two different approaches can be 
identified:

\begin{enumerate}
    \item Free distribution of the individual reservation budgets for a fixed cumulative reservation budget.
    \item Fixed reservation budget distribution, whilst increasing the amount of reservations and thus decreasing the individual budgets.
\end{enumerate}

The first approach is illustrate in the following example.
\begin{example}
\label{ex:implicit-dag}
Let $\tau_i$ be an implicit-deadline, sporadic DAG task with worst-case execution-time $C_i=8$, critical-path length $L_i=5$, 
period $T_i=7$ and relative deadline $D_i=7$.
In order to minimize the cumulative reservation budget as given by Eq.~\eqref{eq:mi-lowerbound}, it is 
mandatory to minimize the number of reservation servers $m_i$. 
The smallest $m_i$ that satisfies Eq.~\eqref{eq:mi-lowerbound} is given by
\begin{align}
  \label{eq:example-cumulative-bound}
  \ceiling{\frac{C_i-L_i}{T_i-L_i}} = \ceiling{\frac{8-5}{7-5}}=2
\end{align} and implies that the largest possible budget, i.e., the tasks
relative deadline, is selected.
Therefore the smallest cumulative service, that the two reservation servers need to provide is given by 
$8+5\cdot(2-1) = 13$.
Using the budget constraints, $5 < E_{i,1} \leq 7$ and $5 < E_{i,2} \leq 7$, any
combination of the reservation budgets from $(E_{i,1} = 6, E_{i,2}=7)$
up to $(E_{i,1} = 6.5, E_{i,2} = 6.5)$ suffices the necessary conditions, whilst using the same amount of 
reservation servers.
\end{example}

The benefit of the combination $E_{i,1} = 6$ and $E_{i,2} = 7$ is that one of
them has a smaller execution time at a price that one of them has a higher
execution time. It may be possible that such a combination is easier
to be schedulable, but there is no greedy and simple criteria to find
the most suitable combination in the \emph{global perspective for all the tasks}.

The second approach is illustrated in the following example.
\begin{example}
\label{ex:increase-mi}
Let the task be the same as in Example~\ref{ex:implicit-dag} and let $m_i \geq
2$, then the reservation budgets are set to
\begin{align}
E_i(m_i) = L_i + \frac{C_i-L_i}{m_i} = 5 + \frac{3}{m_i}
\end{align} for all $1 \leq i \leq m_i$.
\end{example}

The benefit of this approach is that, if $E_{i}(m_i)$ is too large to fit on any
processor, $m_i + 1$ reservations with decreased budgets could be scheduled
easier. 

\section{Scheduling Reservation Servers}
\label{sec:ordinary-scheduling}

\subsection{Partitioned Scheduling}

When considering arbitrary-deadline task systems, the exact schedulability test
evaluates the worst-case response time 
using time-demand analysis and a \emph{busy-window}
concept~\cite{DBLP:conf/rtss/Lehoczky90}.
The finishing time  $R_{k,h}$ of the $h$-th job of task $\tau_k$ can be
calculated by finding the minimum $t$ in the busy window where 
\begin{equation}
h E_k + \sum_{\tau_i \in {\tau}_m} \ceiling{\frac{t}{T_i}}E_i \leq t.
\end{equation}
This means, the response time of the $h$-th job is  
\mbox{$R_{k,h}-(h-1)T_k$}. 
If $R_{k,h} \leq h T_k$, the busy window of task $\tau_k$
finishes with the $h$-th job. Therefore, the worst-case
response time of $\tau_k$ is the maximum
response time among the jobs in the busy
window~\cite{DBLP:conf/rtss/Lehoczky90}.  
While this provides an exact schedulability test, 
the test has an exponential time complexity since 
the length of the busy window can be up to the task sets hyper-period which is 
exponential with respect to the input size.


Fisher, Baruah, and Baker~\cite{DBLP:conf/ecrts/FisherBB06} provided the
following approximated test: 
\begin{subequations}
\begin{align}
   E_k + \sum_{\tau_i \in {\bf
    T}_m} \left(1+\frac{D_k}{T_i}\right)E_i \leq D_k & \;\;\;\mbox{ and }   \label{eq:fbb-arbitrary-1}\\
  U_k + \sum_{\tau_i \in {\bf
    T}_m} U_i \leq 1\label{eq:fbb-arbitrary-2}
\end{align}  
\end{subequations}
Eq.~\eqref{eq:fbb-arbitrary-2} ensures that the workload after $D_k$ is not
underestimated  when arbitrary deadline task
systems are considered, which could happen in  Eq.~\eqref{eq:fbb-arbitrary-1}.


Bini et al.~\cite{DBLP:journals/tc/BiniNRB09} improved the analysis
in~\cite{DBLP:conf/ecrts/FisherBB06} by providing a tighter analysis than
Eq.~\eqref{eq:fbb-arbitrary-1}, showing that  
the worst-case response time of task $\tau_k$ is at most 
\[
\frac{E_k+ \sum_{\tau_i \in {\tau}_m} E_i - \sum_{\tau_i \in {\bf
      T}_m } U_i E_i}{1-\sum_{\tau_i \in {\tau}_m} U_i}.
\]
Therefore, the schedulability condition in
Eqs. \eqref{eq:fbb-arbitrary-1} and \eqref{eq:fbb-arbitrary-2}
can be rewritten as
\begin{subequations}
\begin{align}
  E_k + D_k(\sum_{\tau_i \in {\tau}_m} U_i)+ \sum_{\tau_i \in {\bf 
      T}_m} E_i - \sum_{\tau_i \in {\tau}_m } U_i E_i&\leq 
  D_k  \label{eq:fbb-arbitrary-3}\\
 U_k + \sum_{\tau_i \in {\bf 
    T}_m} U_i &\leq 1\label{eq:fbb-arbitrary-4}
\end{align}    
\end{subequations}

\subsection{Competitiveness}
\label{sec:competitive}

This section will analyze the
theoretical properties when scheduling the reservation servers based on the
deadline-monotonic (DM) partitioning strategy. It has been proved by Chen~\cite{ChenECRTS2016-Partition}
that such a strategy has a speedup factor of $2.84306$ \mbox{(respectively,
$3$)}  against the optimal schedule for ordinary constrained-deadline (respectively, arbitrary-deadline) task systems when the
fixed-priority deadline-monotonic scheduling algorithm is
used. Moreover, Chen et
al.~\cite{Chakraborty2011a,DBLP:journals/rts/ChenC13} also showed that
such a strategy has a speedup factor of $2.6322$ \mbox{(respectively, $3$)}
against the optimal schedule for ordinary constrained-deadline
(respectively, arbitrary-deadline) task systems when the
dynamic-priority earliest-deadline-first (EDF) scheduling algorithm is
used.


\begin{theorem}
  \label{thm:gamma-inflation}
  Suppose that $\gamma>1$ is given, \mbox{$C_i > L_i$},
  and there are
  exactly $m_i$ reservation servers for task $\tau_i$ where \mbox{$m_i
  \defeq \ceiling{\frac{C_i-L_i}{L_i(\gamma-1)}}$} with $m_i \geq 2$. If
  \mbox{$C_i' = \sum_{j=1}^{m_i} E_{i,j} = C_i + (m_i-1) \cdot L_i$}, then $C_i' \leq
  (1+\frac{1}{\gamma-1}) \cdot C_i$.
\end{theorem}
\begin{proof}
  By the assumption $L_i > 0$ and $\gamma > 1$, the setting of
  $m_i=\ceiling{\frac{C_i-L_i}{L_i(\gamma-1)}}$ implies that
  {\footnotesize{\begin{align}
    &m_i-1 < \frac{C_i-L_i}{L_i(\gamma-1)} \leq m_i \label{eq:m_i-floor}\\
    \Rightarrow &(m_i-1)(\gamma-1)L_i < C_i-L_i \leq m_i
    L_i\gamma  - m_iL_i \label{eq:m_i-floor-expand}\\
    \Rightarrow     &  C_i + (m_i-1)L_i \leq m_i \gamma L_i < C_i +
    (m_i -2) L_i +\gamma L_i \label{eq:m_i-floor-final}
  \end{align}}}
  The condition $m_i \gamma L_i < C_i +
    (m_i -2) L_i +\gamma L_i$ in Eq.~\eqref{eq:m_i-floor-final} implies $(m_i -1) L_i < \frac{C_i +
    (m_i -2) L_i}{\gamma}$ since $\gamma > 0$.
  Since \mbox{$C_i' = \sum_{j=1}^{m_i} E_{i,j} = C_i + (m_i-1) L_i$} by
  definition, we know 
  {\small  \begin{align*}
      C_i'  < \qquad&C_i + \frac{C_i
        +(m_i-2)L_i}{\gamma}\\
      <_1 \qquad& \frac{C_i(\gamma+1)}{\gamma} + \frac{(m_i-2)}{\gamma}\left(\frac{C_i}{(m_i-1)(\gamma-1)}\right)\\
      \leq_2\qquad & C_i\left(\frac{\gamma+1}{\gamma} +
        \frac{1}{\gamma^2-\gamma}\right)\\
      = \qquad &  \left(1+\frac{1}{\gamma-1}\right) \cdot C_i
    \end{align*}}
    where $<_1$ is due to $L_i <\frac{C_i}{(m_i-1)(\gamma-1)+1} < \frac{C_i}{(m_i-1)(\gamma-1)}$ by
      reorganizing the condition in Eq.~\eqref{eq:m_i-floor} and
      $\leq_2$ is due to $m_i \geq 2$ and $\frac{m_i-2}{m_i-1} \leq 1$.
\end{proof}

\begin{lemma}
\label{lemma:relation-individual}
Under the same setting as in Theorem~\ref{thm:gamma-inflation},
\begin{equation}
  \label{eq:relation-individual}
\frac{C_i'}{m_i}  = \frac{C_i + (m_i-1)L_i}{m_i} \leq \gamma L_i  
\end{equation}
\end{lemma}
\begin{proof}
  \begin{align*}
    \frac{C_i + (m_i-1)L_i}{m_i} = & L_i + \frac{C_i - L_i}{m_i} = L_i + \frac{C_i
      - L_i}{\ceiling{\frac{C_i-L_i}{L_i(\gamma-1)}}}\\
    \leq & L_i + \frac{C_i
      - L_i}{\frac{C_i-L_i}{L_i(\gamma-1)}} = \gamma L_i
  \end{align*}
where the inequality is due to $C_i > L_i$ and $\gamma > 1$.
\end{proof}

\begin{algorithm}[t]
\caption{R-EQUAL Algorithm}
\label{alg:ert}
\begin{algorithmic}[1]
\STATE ${\tau}_{\textsc{heavy}} \leftarrow \{\tau_i \in {\tau} \ | \ C_i > \gamma L_i\}$
\STATE ${\tau}_{\textsc{light}} \leftarrow \{\tau_i \in {\tau} \ | \   C_i \leq \gamma L_i\}$
\STATE ${\tau} \leftarrow {\tau}_{\textsc{light}}$
\FOR{each $\tau_i \in {\tau}_{\textsc{heavy}}$}
\STATE $m_i \leftarrow \ceiling{\frac{C_i-L_i}{L_i(\gamma-1)}}$
\FOR{$1 \leq j \leq m_i$}
\STATE $E_{i,j} \leftarrow \frac{C_i + (m_i-1)L_i}{m_i}$
\STATE $\tau_{i,j} \leftarrow  (E_{i,j}, D_i, T_i)$
\STATE ${\tau} \leftarrow {\tau}^{*} \cup \setof{\tau_{i,j}}$
\ENDFOR
\ENDFOR
\RETURN ${\tau}$
\end{algorithmic}
\end{algorithm}

The result in Theorem~\ref{thm:gamma-inflation} can be used to specify
an algorithm that transforms a collection of sporadic, arbitrary
deadline DAG tasks into a transformed collection of light sporadic
reservation tasks with a \emph{constant $\gamma$}, illustrated in
Algorithm~\ref{alg:ert}.  The algorithm simply classifies a task
$\tau_i$ as a heavy task if $C_i > \gamma L_i$ and a light task if
$C_i \leq \gamma L_i$, respectively. If task $\tau_i$ is a heavy task,
$m_i \leftarrow \ceiling{\frac{C_i-L_i}{L_i(\gamma-1)}}$ reservation
servers will be provided, each with an execution time budget of
$\frac{C_i + (m_i-1)L_i}{m_i}$.

We implicitly assume $L_i
\leq D_i$ in Algorithm~\ref{alg:ert}. After the transformation, we can
apply any existing scheduling algorithms for scheduling ordinary
sporadic real-time task systems to partition or schedule the
reservation servers.

\begin{lemma}
  \label{lemma-correct-classification}
  By adopting Algorithm~\ref{alg:ert}, for a given $\gamma > 1$, 
  \begin{itemize}
  \item if a task $\tau_i$ is in ${\tau}_{\textsc{heavy}}$, $m_i \geq
    2$, Theorem~\ref{thm:gamma-inflation} holds, and $E_{i,j} =
    \frac{C_i + (m_i-1)L_i}{m_i}$ for
    $j=1,2,\ldots,m_i$;
  \item if a task $\tau_i$ is in ${\tau}_{\textsc{light}}$, $m_i = 1$,
    and $\tau_i$ is executed sequentially without any inflation of
    execution time, i.e., $E_{i,1} = C_i$.
  \end{itemize}
  Furthermore, $E_{i,j}
  \leq \gamma L_i$ for any $j=1,2,\ldots,m_i$, and \mbox{$C_i' =
  \sum_{j=1}^{m_i} E_{i,j} = C_i + (m_i-1) \cdot L_i \leq
  (1+\frac{1}{\gamma-1}) \cdot C_i$} for both light and heavy tasks.
\end{lemma}
\begin{proof}
  This holds according to the above discussions in
  Theorem~\ref{thm:gamma-inflation} and Lemma~\ref{lemma:relation-individual}.
\end{proof}


\begin{theorem}
  \label{thm:arbitrary-DM}
  A system of arbitrary-deadline DAG tasks scheduled by
  {reservation-based federated scheduling} under partitioned DM
  admits a constant speedup factor of $3+2\sqrt{2}$ with respect to
  any optimal scheduler by setting $\gamma$ to $1+\sqrt{2}$.
\end{theorem}
\begin{proof}
  We first adopt Algorithm~\ref{alg:ert} with a setting of
  \mbox{$\gamma=1+\sqrt{2}$}.  If there exists a DAG task in which
  \mbox{$(1+\sqrt{2})L_i > D_i$}, then we know that the speedup factor for
  this task set is $(1+\sqrt{2})$. We focus on the case that $\gamma
  L_i \leq D_i$.

  Suppose that $\tau_{k,\ell}$ is a reservation task that is not able
  to be partitioned to any of the given $M$ processors, where \mbox{$1 \leq
  \ell \leq m_k$}.  Let ${\bf M}_1$ be the set of processors in which
  Eq.~\eqref{eq:fbb-arbitrary-1} fails. Let ${\bf M}_2$ be the set of
  processors in which Eq.~\eqref{eq:fbb-arbitrary-1} succeeds but
  Eq.~\eqref{eq:fbb-arbitrary-2} fails. Since $\tau_{k,\ell}$ cannot
  be assigned on any of the $M$ processors 
  \mbox{$|{\bf M}_1|+|{\bf
    M}_2|=M$.} By the violation of Eq.~\eqref{eq:fbb-arbitrary-1}, we
  know that
  \begin{align}
    &  |{\bf M}_1|E_{k,\ell} + \sum_{m \in {\bf M}_1}\sum_{\tau_{i,j} \in {\bf
        T}_m}\left(1+\frac{D_k}{T_i}\right)E_{i,j} > |{\bf
      M}_1|D_k    \nonumber\\
    \Rightarrow & |{\bf M}_1|\frac{E_{k,\ell}}{D_k} + \sum_{m \in {\bf M}_1}\sum_{\tau_{i,j} \in {\bf
        T}_m}\left(\frac{E_{i,j}}{D_k}+\frac{E_{i,j}}{T_i}\right) > |{\bf
      M}_1| \label{eq:fbb-arbitrary-violate-1}
  \end{align}
  By the violation of Eq.~\eqref{eq:fbb-arbitrary-2}, we
  know that
  \begin{equation}
    \label{eq:fbb-arbitrary-violate-2}
    |{\bf M}_2|\frac{E_{k,\ell}}{T_k} + \sum_{m \in {\bf M}_2}\sum_{\tau_{i,j} \in {\bf
        T}_m} \frac{E_{i,j}}{T_{i,j}} > |{\bf M}_2|
  \end{equation}
  By Eqs.~\eqref{eq:fbb-arbitrary-violate-1}
  and~\eqref{eq:fbb-arbitrary-violate-2}, the definition
  $\sum_{j=1}^{m_i}E_{i,j} = C_i'$, and the fact that $\tau_{i,j}$ is
  assigned either on a processor of ${\bf M}_1$ or on a processor of
  ${\bf M}_2$ if $\tau_{i,j}$ is assigned successfully prior to
  $\tau_{k,\ell}$, we know that {\small \begin{align}
      &  M \frac{E_{k,\ell}}{\min\{T_k, D_k\}} + \sum_{i=1}^{k} \left(\frac{C_i'}{T_i} + \frac{C_i'}{D_k}\right) > M
    \end{align}} 
By Lemma~\ref{lemma-correct-classification}, $E_{k,\ell}
  \leq \gamma L_k$ and
   \mbox{$C_i' \leq
  (1+\frac{1}{\gamma-1}) C_i$,} the above inequality implies also  
  {\small
  \begin{align} & M \frac{\gamma L_k}{\min\{T_k, D_k\}} + \sum_{i=1}^{k} \left(\frac{(1+\frac{1}{\gamma-1}) C_i}{T_i} +
        \frac{(1+\frac{1}{\gamma-1})C_i}{D_k}\right) > M
    \end{align}
  } Let $X$ be $\max\left\{\frac{L_k}{\min\{T_k, D_k\}},
    \sum_{i=1}^{k} \frac{C_i}{M T_i}, \sum_{i=1}^{k} \frac{C_i}{M
      D_k}\right\}$. Therefore, we know that\footnote{The setting of $\gamma$ as $1+\sqrt{2}$ is in fact to maximize $\frac{\gamma-1}{\gamma^2+\gamma}$.}
  \begin{align}
&    \gamma X + 2 \left(1+\frac{1}{\gamma-1}\right)X > 1    \\
\Rightarrow\quad & X > \frac{1}{\gamma + \frac{2\gamma}{\gamma-1}} = \frac{\gamma-1}{\gamma^2+\gamma} = \frac{\sqrt{2}}{4+3\sqrt{2}} = \frac{1}{3+2\sqrt{2}}
  \end{align}
  Since $D_i \leq D_k$ under deadline-monotonic partitioning, we know
  that the task system is not schedulable at speed $X$. Therefore, the
  speedup factor of the reservation-based federated Scheduling is at most $3+2\sqrt{2}$.
\end{proof}

\begin{theorem}
  \label{thm:arbitrary-EDF}
  A system of arbitrary-deadline DAG tasks scheduled by
  \textbf{reservation-based federated scheduling} under partitioned EDF
  admits a constant speedup factor of $3+2\sqrt{2}$ with respect to
  any optimal scheduler by setting $\gamma$ to $1+\sqrt{2}$.
\end{theorem}
\begin{proof}
  Since EDF is an optimal uniprocessor scheduling policy with respect
  to schedulability, the same task partitioning algorithm and analysis
  used in Theorem~\ref{thm:arbitrary-DM} yield the result directly.
\end{proof}

\subsection*{Future research}
We will design concrete algorithms, that create 
\emph{non-equal} reservation budgets and compare their competitiveness 
against the \emph{R-EQUAL} algorithm. Further we want to analyse the 
performance of the proposed reservation based DAG task scheduling in 
global scheduling algorithms.
Finally the incorporation of self-suspending behaviour of the reservation 
servers may yield analytic and practical benefits, since in our current 
approach the \emph{worst-case} DAG task structure has to be assumed in 
order to provide provably sufficient resources. This is often too 
pessimistic and self-suspending behaviour can potentially help 
to service the \emph{actual} demands more precisely without spinning 
 and blocking resources unused.
\bibliographystyle{plain}
\bibliography{rtbib,real-time} 
\end{document}